%% file: MBPS-MainFile.tex
\documentclass[runningheads]{llncs}
\usepackage[T1]{fontenc}
\usepackage{graphicx}
\usepackage{amsmath}

\usepackage{algorithm}

\usepackage{algpseudocode}

\usepackage{multirow}

\usepackage{amssymb} 

\usepackage{tabularx}

\usepackage{array} 

\newcolumntype{C}{>{\centering\arraybackslash}X}

\usepackage{cite}
\usepackage{hyperref}

\usepackage{subcaption} 
\usepackage{tikz} 

\usepackage{newfloat}
\usepackage{caption}
\DeclareFloatingEnvironment[placement={!ht},name=Graph]{Graph}

\usepackage{chngcntr}
\counterwithout{Graph}{section}

\usepackage{pgfplots}

\usepgfplotslibrary{groupplots}

\input{macros}

\usepackage{todonotes}

\begin{document}
\title{Unleashing Multicore Strength for Efficient Execution of Blockchain Transactions}
%
%

\titlerunning{Unleashing Multicore Strength for Efficient Execution of Transactions}

\author{Ankit Ravish\inst{1} \and
Akshay Tejwani\inst{1} \and
Piduguralla Manaswini\inst{1} \and
Sathya Peri\inst{1}}
\authorrunning{Ankit et al.}
%
\institute{Indian Institute of Technology Hyderabad, Kandi, Telangana 502285
\email{\{cs21resch11014,cs21mtech12015,cs20resch11007\}@iith.ac.in, sathya\_p@cse.iith.ac.in}}
\maketitle              
\begin{abstract}

Blockchain technology is booming up the digital world in recent days and thus paved a way for creating separate blockchain network for various industries.
This technology is characterized by its distributed, decentralized, and immutable ledger system which serves as a fundamental platform for managing smart contract transactions (SCTs).
However, these self-executing codes implemented using blockchains undergo sequential validation within a block which introduces performance bottlenecks.
In response, this paper introduces a framework called the Multi-Bin Parallel Scheduler (MBPS) designed for parallelizing blockchain smart contract transactions to leverage the capabilities of multicore systems.
Our proposed framework facilitates concurrent execution of SCTs, enhancing performance by allowing non-conflicting transactions to be processed simultaneously while preserving deterministic order.
The framework comprises of three vital stages: conflict detection, bin creation and execution.
We conducted an evaluation of our MBPS framework in Hyperledger Sawtooth v1.2.6, revealing substantial performance enhancements compared to existing parallel SCT execution frameworks across various smart contract applications.
This research contributes to the ongoing optimization efforts in blockchain technology demonstrating its potential for scalability and efficiency in real-world scenarios.

\end{abstract}

\section{Introduction}
\label{sec:introduction}
\input{section-folder/Introduction}

\section{Motivation}
\label{sec:motivation}
\input{section-folder/Motivation}

\section{Related Work}
\label{sec:related}
\input{section-folder/Related-Work}

\section{Proposed Framework}
\label{sec:framework}
\input{section-folder/Proposed-Framework}

\section{Analysis of Experiments}
\label{sec:experiments}
\input{section-folder/Experiments}

\section{Conclusion and Future Work}
\label{sec:conclusion}
\input{section-folder/Conclusion-Future-Work}

%
%

\bibliographystyle{splncs04}
\bibliography{citations}

\end{document}

%% file: macros.tex
\newcommand{\cmnt}[1]{}
\newcommand{\ignore}[1]{}

\newcommand{\tabref}[1]{Table~\ref{tab:#1}}

\newcommand{\grfref}[1]{Graph~\ref{graph:#1}}

\newcommand{\remove}[1]{}



%% file: section-folder/Introduction.tex
Blockchain \cite{article} is an innovative and decentralized digital ledger technology that facilitates secure, transparent, and tamper-resistant record-keeping. Unlike traditional centralized databases, a blockchain comprises a chain of blocks, each containing a set of transactions. These blocks are linked together using cryptographic hashes, ensuring data integrity and immutability. This distributed ledger system relies on a consensus mechanism among network participants, making it a trustless and open platform for various applications.

The contemporary world is experiencing a digital revolution, marked by an explosive growth in data and digital transactions. The necessity of blockchain in today's world is underscored by its numerous advantages. Traditional centralized systems are vulnerable to fraud, cyberattacks, and manipulation. Blockchain addresses these issues and ensures data integrity, reduces fraud, and eliminates intermediaries in financial transactions, making it a cornerstone for financial institutions \cite{inbook}. Additionally, it enables secure supply chain management, simplifies identity verification, and ensures the integrity of critical records \cite{tapscott2016blockchain}.

Smart contracts, another pillar of blockchain technology, represent self-executing agreements with the terms of the contract directly written into code. These contracts run on blockchain technology and are designed to automatically execute when predefined conditions are met, eliminating the need for intermediaries and reducing the risk of fraud and manipulation \cite{journals/firstmonday/Szabo97}.
Smart contracts reduce the chances of disputes by providing a transparent, automated, and tamper-proof mechanism for executing agreements. They also significantly reduce transaction costs and the time required for contract execution.

Blockchain operates through a network of nodes (computers) that work collaboratively to validate and record transactions in a chronological sequence of blocks. The process of validating transactions and smart contracts begins with a block producer, who gathers a group of pending transactions and attempts to create a new block. However, before this block is added to the blockchain, it must undergo a rigorous validation process by all participating nodes \cite{antonopoulos2014mastering, wood2014ethereum}.

Smart contracts within a blockchain network are executed serially, meaning one contract is processed at a time. This sequential execution ensures that the blockchain maintains a single, consistent state at any given moment but can lead to bottlenecks, limiting the system's capacity to handle a large number of transactions simultaneously. This bottleneck results from the serial nature of smart contract execution, causing delays and potentially hindering scalability.
Centralized servers, as seen in traditional financial systems, can process thousands of transactions per second, enabling high-frequency trading and efficient payment processing. For instance, Visa, a centralized payment network, claims to handle over 24,000 TPS. In contrast, decentralized blockchain networks often face scalability challenges, limiting their TPS. For example, Bitcoin \cite{bitcoin} has a TPS capacity of approximately 7-10 transactions per second, while Ethereum's \cite{ethereum} throughput is around 30 TPS.


%% file: section-folder/Motivation.tex
Over the past few years, blockchain technology has witnessed a surge in popularity, gaining widespread adoption across diverse domains. However, a significant challenge that has arisen is the issue of scalability.
Parallel execution of smart contract transactions offers a compelling solution to the bottlenecks inherent in serial execution. By allowing multiple transactions to be processed simultaneously, blockchain networks can significantly enhance their throughput and scalability. It reduces transaction confirmation times, making blockchain applications more responsive. It also optimizes resource utilization and enhances the overall performance of blockchain networks, making them more efficient and adaptable to the demands of modern applications \cite{Buterin2013}.

The solution to this problem, which involves the parallel execution of smart contract transactions, is highly complex.
When a block producer and validators execute transactions in parallel, there exists a risk that validators may execute conflicting transactions in a different order than the block producer. Consequently, the serialization order implemented by the block producer might differ from that of the validator, potentially leading to the validator arriving at a final state distinct from that of the block producer. This discrepancy could result in the incorrect rejection of a valid block by the validator.

\vspace{-0.2cm}

\input{figure-folder/figure-conflicting-address}

To illustrate this challenge, consider Figures \ref{fig:no_conflict_address} and \ref{fig:conflict_address}. Figure \ref{fig:no_conflict_address} shows two non-conflicting transactions, $T_1$ and $T_2$, operating on different shared data objects. In this scenario, the execution order between $T_1$ and $T_2$ is insignificant. Conversely, Figure \ref{fig:conflict_address} illustrates two conflicting transactions, $T_1$ and $T_2$, working on the same shared data objects. In this case, the validators must maintain the same order of execution as the block producer to avoid distinct final states. For instance, with an initial state of $A$ and $B$ set at 10, parallel execution 1 could yield a final state of $A$ as 20 and $B$ as 0, while parallel execution 2 could result in a final state of $A$ as 0 and $B$ as 20.

Such potential discrepancies in final states pose a risk of valid blocks being incorrectly rejected by the validators, which is an undesirable outcome. Therefore, maintaining transaction order consistency among validators during parallel execution is imperative to ensure the accurate validation of blocks. This issue underscores the need for robust mechanisms to synchronize the execution of transactions across the blockchain network, preventing unintended variations in the final state and ensuring the integrity of the validation process.

Thus, this paper introduces the Multi-Bin Parallel Scheduler (MBPS) framework, specifically designed to schedule the parallel execution of smart contracts, thereby addressing the aforementioned issues.


\noindent Our contributions to this paper are as follows:
\vspace{-0.6em}
\begin{itemize}
    \item We provided a comprehensive overview of related work (Section \ref{sec:related}) that aligns with our proposed approach, contextualizing our research within existing studies and methodologies.
    \item We presented detailed implementation aspects of our proposed framework (Section \ref{sec:framework}), including in-depth explanations of the algorithms and the development of three distinct frameworks: Standard MBPS (Section \ref{sec:standard}), Assisted MBPS (Section \ref{sec:assisted}) and Lockfree MBPS (Section \ref{sec:lockfree}). Additionally, we introduced relevant lemmas to support our framework.
    \item We detailed the experimental setup (Section \ref{sec:experiments}) used for evaluation and presented the results, providing a thorough analysis of the methods and outcomes.
\end{itemize}

%% file: figure-folder/figure-conflicting-address.tex
\begin{figure}[htbp]
    \centering
    \begin{subfigure}[b]{0.3\textwidth}
        \centering
        \begin{tikzpicture}
            \draw[color=red] (0,0.2) -- (0,-0.2);
            \draw[color=red] (0,0.2) -- (0.1,0.2);
            \draw[color=red] (0,-0.2) -- (0.1,-0.2);
            
            \draw[color=red] (0,0) -- (1.2,0);
            \draw (0.6,0.1) -- (0.6,-0.1);

            \node[left] at (0,0) {$T_{1}$};
            \node[above] at (0.6,0.2) {\tiny \textit{transfer(A,B,10)}};
            \node[right] at (1.2,0) {$C_{1}$};

            \draw[color=red] (1.2,0.2) -- (1.2,-0.2);
            \draw[color=red] (1.2,0.2) -- (1.1,0.2);
            \draw[color=red] (1.2,-0.2) -- (1.1,-0.2);


            \draw[color=blue] (1.5,-0.3) -- (1.5,-0.7);
            \draw[color=blue] (1.5,-0.3) -- (1.6,-0.3);
            \draw[color=blue] (1.5,-0.7) -- (1.6,-0.7);
            
            \draw[color=blue] (1.5,-0.5) -- (2.7,-0.5);
            \draw (2.1,-0.4) -- (2.1,-0.6);

            \node[left] at (1.5,-0.5) {$T_{2}$};
            \node[below] at (2.1,-0.7) {\tiny \textit{transfer(C,D,10)}};
            \node[right] at (2.7,-0.5) {$C_{2}$};

            \draw[color=blue] (2.7,-0.3) -- (2.7,-0.7);
            \draw[color=blue] (2.7,-0.3) -- (2.6,-0.3);
            \draw[color=blue] (2.7,-0.7) -- (2.6,-0.7);
            
        \end{tikzpicture}
        \caption{\scriptsize Sequential Execution}
        \label{fig:seqex-a}
    \end{subfigure}
    \begin{subfigure}[b]{0.3\textwidth}
        \centering
        \begin{tikzpicture}
            \draw[color=red] (0,0.2) -- (0,-0.2);
            \draw[color=red] (0,0.2) -- (0.1,0.2);
            \draw[color=red] (0,-0.2) -- (0.1,-0.2);
            
            \draw[color=red] (0,0) -- (1.2,0);
            \draw (0.6,0.1) -- (0.6,-0.1);

            \node[left] at (0,0) {$T_{1}$};
            \node[above] at (0.6,0.2) {\tiny \textit{transfer(A,B,10)}};
            \node[right] at (1.2,0) {$C_{1}$};

            \draw[color=red] (1.2,0.2) -- (1.2,-0.2);
            \draw[color=red] (1.2,0.2) -- (1.1,0.2);
            \draw[color=red] (1.2,-0.2) -- (1.1,-0.2);


            \draw[color=blue] (0.4,-0.3) -- (0.4,-0.7);
            \draw[color=blue] (0.4,-0.3) -- (0.5,-0.3);
            \draw[color=blue] (0.4,-0.7) -- (0.5,-0.7);
            
            \draw[color=blue] (0.4,-0.5) -- (1.6,-0.5);
            \draw (1,-0.4) -- (1,-0.6);

            \node[left] at (0.4,-0.5) {$T_{2}$};
            \node[below] at (1,-0.7) {\tiny \textit{transfer(C,D,10)}};
            \node[right] at (1.6,-0.5) {$C_{2}$};

            \draw[color=blue] (1.6,-0.3) -- (1.6,-0.7);
            \draw[color=blue] (1.6,-0.3) -- (1.5,-0.3);
            \draw[color=blue] (1.6,-0.7) -- (1.5,-0.7);
        \end{tikzpicture}
        \caption{\scriptsize Parallel Execution 1}
        \label{fig:parex1-a}
    \end{subfigure}
    \begin{subfigure}[b]{0.3\textwidth}
        \centering
        \begin{tikzpicture}
            \draw[color=red] (0.4,0.2) -- (0.4,-0.2);
            \draw[color=red] (0.4,0.2) -- (0.5,0.2);
            \draw[color=red] (0.4,-0.2) -- (0.5,-0.2);
            
            \draw[color=red] (0.4,0) -- (1.6,0);
            \draw (1,0.1) -- (1,-0.1);

            \node[left] at (0.4,0) {$T_{1}$};
            \node[above] at (1,0.2) {\tiny \textit{transfer(A,B,10)}};
            \node[right] at (1.6,0) {$C_{1}$};

            \draw[color=red] (1.6,0.2) -- (1.6,-0.2);
            \draw[color=red] (1.6,0.2) -- (1.5,0.2);
            \draw[color=red] (1.6,-0.2) -- (1.5,-0.2);


            \draw[color=blue] (0,-0.3) -- (0,-0.7);
            \draw[color=blue] (0,-0.3) -- (0.1,-0.3);
            \draw[color=blue] (0,-0.7) -- (0.1,-0.7);
            
            \draw[color=blue] (0,-0.5) -- (1.2,-0.5);
            \draw (0.6,-0.4) -- (0.6,-0.6);

            \node[left] at (0,-0.5) {$T_{2}$};
            \node[below] at (0.6,-0.7) {\tiny \textit{transfer(C,D,10)}};
            \node[right] at (1.2,-0.5) {$C_{2}$};

            \draw[color=blue] (1.2,-0.3) -- (1.2,-0.7);
            \draw[color=blue] (1.2,-0.3) -- (1.1,-0.3);
            \draw[color=blue] (1.2,-0.7) -- (1.1,-0.7);
        \end{tikzpicture}
        \caption{\scriptsize Parallel Execution 2}
        \label{fig:parex2-a}
    \end{subfigure}
    \caption{ Sequential and Parallel Execution (No Conflicting Addresses)}
    \label{fig:no_conflict_address}
\end{figure}
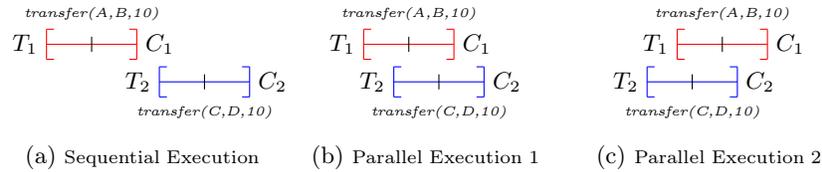

\vspace{-1cm}

\begin{figure}[htbp]
    \centering
    \begin{subfigure}[b]{0.3\textwidth}
        \centering
        \begin{tikzpicture}
            \draw[color=red] (0,0.2) -- (0,-0.2);
            \draw[color=red] (0,0.2) -- (0.1,0.2);
            \draw[color=red] (0,-0.2) -- (0.1,-0.2);
            
            \draw[color=red] (0,0) -- (1.2,0);
            \draw (0.6,0.1) -- (0.6,-0.1);

            \node[left] at (0,0) {$T_{1}$};
            \node[above] at (0.6,0.2) {\tiny \textit{transfer(A,B,10)}};
            \node[right] at (1.2,0) {$C_{1}$};

            \draw[color=red] (1.2,0.2) -- (1.2,-0.2);
            \draw[color=red] (1.2,0.2) -- (1.1,0.2);
            \draw[color=red] (1.2,-0.2) -- (1.1,-0.2);

            \draw[color=blue] (1.5,-0.3) -- (1.5,-0.7);
            \draw[color=blue] (1.5,-0.3) -- (1.6,-0.3);
            \draw[color=blue] (1.5,-0.7) -- (1.6,-0.7);
            
            \draw[color=blue] (1.5,-0.5) -- (2.7,-0.5);
            \draw (2.1,-0.4) -- (2.1,-0.6);

            \node[left] at (1.5,-0.5) {$T_{2}$};
            \node[below] at (2.1,-0.7) {\tiny \textit{transfer(B,A,10)}};
            \node[right] at (2.7,-0.5) {$C_{2}$};

            \draw[color=blue] (2.7,-0.3) -- (2.7,-0.7);
            \draw[color=blue] (2.7,-0.3) -- (2.6,-0.3);
            \draw[color=blue] (2.7,-0.7) -- (2.6,-0.7);
            
        \end{tikzpicture}
        \caption{\scriptsize Sequential Execution}
        \label{fig:seqex-b}
    \end{subfigure}
    \begin{subfigure}[b]{0.3\textwidth}
        \centering
        \begin{tikzpicture}
            \draw[color=red] (0,0.2) -- (0,-0.2);
            \draw[color=red] (0,0.2) -- (0.1,0.2);
            \draw[color=red] (0,-0.2) -- (0.1,-0.2);
            
            \draw[color=red] (0,0) -- (1.2,0);
            \draw (0.6,0.1) -- (0.6,-0.1);

            \node[left] at (0,0) {$T_{1}$};
            \node[above] at (0.6,0.2) {\tiny \textit{transfer(A,B,10)}};
            \node[right] at (1.2,0) {$C_{1}$};

            \draw[color=red] (1.2,0.2) -- (1.2,-0.2);
            \draw[color=red] (1.2,0.2) -- (1.1,0.2);
            \draw[color=red] (1.2,-0.2) -- (1.1,-0.2);


            \draw[color=blue] (0.4,-0.3) -- (0.4,-0.7);
            \draw[color=blue] (0.4,-0.3) -- (0.5,-0.3);
            \draw[color=blue] (0.4,-0.7) -- (0.5,-0.7);
            
            \draw[color=blue] (0.4,-0.5) -- (1.6,-0.5);
            \draw (1,-0.4) -- (1,-0.6);

            \node[left] at (0.4,-0.5) {$T_{2}$};
            \node[below] at (1,-0.7) {\tiny \textit{transfer(B,A,10)}};
            \node[right] at (1.6,-0.5) {$C_{2}$};

            \draw[color=blue] (1.6,-0.3) -- (1.6,-0.7);
            \draw[color=blue] (1.6,-0.3) -- (1.5,-0.3);
            \draw[color=blue] (1.6,-0.7) -- (1.5,-0.7);
        \end{tikzpicture}
        \caption{\scriptsize Parallel Execution 1}
        \label{fig:parex1-b}
    \end{subfigure}
    \begin{subfigure}[b]{0.3\textwidth}
        \centering
        \begin{tikzpicture}
            \draw[color=red] (0.4,0.2) -- (0.4,-0.2);
            \draw[color=red] (0.4,0.2) -- (0.5,0.2);
            \draw[color=red] (0.4,-0.2) -- (0.5,-0.2);
            
            \draw[color=red] (0.4,0) -- (1.6,0);
            \draw (1,0.1) -- (1,-0.1);

            \node[left] at (0.4,0) {$T_{1}$};
            \node[above] at (1,0.2) {\tiny \textit{transfer(A,B,10)}};
            \node[right] at (1.6,0) {$C_{1}$};

            \draw[color=red] (1.6,0.2) -- (1.6,-0.2);
            \draw[color=red] (1.6,0.2) -- (1.5,0.2);
            \draw[color=red] (1.6,-0.2) -- (1.5,-0.2);


            \draw[color=blue] (0,-0.3) -- (0,-0.7);
            \draw[color=blue] (0,-0.3) -- (0.1,-0.3);
            \draw[color=blue] (0,-0.7) -- (0.1,-0.7);
            
            \draw[color=blue] (0,-0.5) -- (1.2,-0.5);
            \draw (0.6,-0.4) -- (0.6,-0.6);

            \node[left] at (0,-0.5) {$T_{2}$};
            \node[below] at (0.6,-0.7) {\tiny \textit{transfer(B,A,10)}};
            \node[right] at (1.2,-0.5) {$C_{2}$};

            \draw[color=blue] (1.2,-0.3) -- (1.2,-0.7);
            \draw[color=blue] (1.2,-0.3) -- (1.1,-0.3);
            \draw[color=blue] (1.2,-0.7) -- (1.1,-0.7);
        \end{tikzpicture}
        \caption{\scriptsize Parallel Execution 2}
        \label{fig:parex2-b}
    \end{subfigure}
    \caption{ Sequential and Parallel Execution (Conflicting Addresses)}
    \label{fig:conflict_address}
\end{figure}
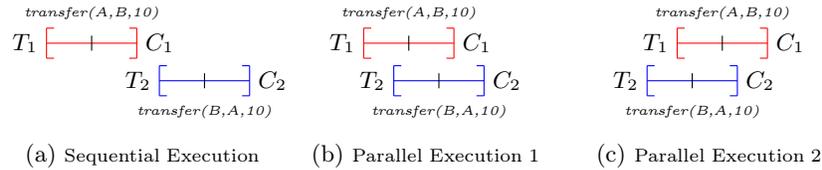

%% file: section-folder/Related-Work.tex
Numerous strategies have been developed to enhance the efficiency and functionality of blockchain technology. These approaches address critical issues such as scalability, security, and consensus mechanisms. These approaches are crucial for addressing the inherent limitations of traditional blockchain networks, which impede the widespread adoption of blockchain technology. One highly effective approach involves utilizing sharding, a technique where the blockchain network is divided into smaller units known as shards, enabling each shard to handle transactions autonomously. This technique significantly increases the throughput of the blockchain by allowing multiple transactions to execute simultaneously across different shards. 

Dickerson et al. \cite{Dickerson+:ACSC:PODC:2017}  introduce an innovative approach for miners and validators to concurrently execute smart contracts, leveraging concepts from software transactional memory. This method involves miners speculatively executing contracts in parallel, enabling non-conflicting transactions to proceed simultaneously. The transactions within a block are organized into a serializable concurrent schedule, represented as a deterministic fork-join program. Validators subsequently re-execute the miner's parallel schedule deterministically yet concurrently based on this encoded program. Saraph et al. \cite{Vikram&Herlihy:EmpSdy-Con:Tokenomics:2019} introduced an Ethereum blockchain algorithm that categorizes transactions into concurrent and sequential bins based on their read and write sets, estimating parallel execution for the former and serial execution for the latter.
The paper \cite{Anjana2020netysObjSC, Anjana2021OptSmart} proposes a framework for concurrent execution of smart contracts using optimistic Software Transactional Memory systems, demonstrating the speedup achieved by concurrent miner and validator over serial counterparts, and providing experimental evaluation on benchmarks from Solidity documentation.

In paper \cite{10.1007/978-3-031-39698-4_13}, the authors have proposed an approach to enhance the performance of blockchain-based smart contract execution through the incorporation of a direct acyclic graph (DAG) based parallel scheduler framework. Liu et al. \cite{10.1109/TPDS.2021.3095234} devised a new approach to smart contract execution, separating consensus nodes from execution nodes to enable parallel transaction execution while allowing consensus nodes to order transactions asynchronously.
Baheti et al. \cite{Baheti2019DiPETrans} presented the DiPETrans framework for distributed transaction execution, where trusted nodes in the blockchain network collaborate in executing transactions and simultaneously perform PoW using a leader-follower approach.

Yan et al. \cite{article-yan} optimized concurrency for each shard, analyzing smart contract features, clustering transactions, and developing a Serializable Schedule and Variable Shadow Speculative Concurrency Control (SCC-VS) algorithm considering factors like transaction frequency, execution time, and conflict rate.
Addressing parallel mode challenges and synchronization issues, \cite{10.1145/3278312.3278321} implemented the proposed model using multi-thread technology and introduced a transaction splitting algorithm to resolve synchronization problems.

%% file: section-folder/Proposed-Framework.tex
This section presents the design of our proposed Multi-Bin Parallel Scheduler (MBPS) framework, detailing its architecture and key components.
The proposed MBPS framework facilitates parallel transaction execution while preserving a deterministic order, thereby leveraging the capabilities of multicore systems to enhance the efficiency of blockchain ecosystems.
This framework introduces three distinct versions to address specific aspects of smart contract execution in blockchain ecosystems: Standard MBPS, Assisted MBPS, and Lockfree MBPS.

The first variant, referred to as \textbf{Standard MBPS}, incorporates a barrier in its design but operates without the assistance of helper threads. This configuration provides a baseline approach with a controlled synchronization point, which allows organized parallel execution of smart contracts.
In contrast, the second variant, \textbf{Assisted MBPS}, employs a barrier in its structure while incorporating helper threads. This variant enhances the standard framework by introducing additional threads that assist in optimizing the execution process when some threads are very slow.
The third variant, \textbf{Lockfree MBPS}, diverges from the barrier-based approach and operates without a synchronization barrier. This barrier-free design and the inclusion of helper threads promote a lock-free execution environment. The efficacy of this variant becomes particularly evident when certain threads experience delays, as the absence of a barrier eliminates the need to wait for other threads to join before progressing to subsequent stages.
\tabref{mbps_frameowrk_comparison} comprehensively compares the different MBPS frameworks, highlighting their distinctive features and characteristics.


\begin{table}[h]
\centering
\caption{Comparison of Different MBPS Framework}
\label{tab:mbps_frameowrk_comparison}
\renewcommand{\arraystretch}{1.6}
\begin{tabularx}{0.8 \textwidth}{||C||C|C||}
\hline
\textbf{Framework} & \textbf{Barrier Free} & \textbf{Helper Threads} \\
\hline
\hline
\textbf{Standard MBPS} & $\times$ & $\times$ \\
\hline
\textbf{Assisted MBPS} & $\times$ & \checkmark \\
\hline
\textbf{Lockfree MBPS} & \checkmark & \checkmark \\
\hline
\end{tabularx}
\end{table}


Each MBPS framework undergoes three crucial stages: conflict detection, bin assignment and execution.

\textbf{Conflict Detection:} In this stage, the objective is to identify conflicts between transactions to assign them to the correct bins. Two transactions, denoted as $T_i$ and $T_j$, are said to be in conflict if any of the following conditions hold:
\begin{itemize}
\item $T_i$ reads a data item, and $T_j$ writes to the same data item
\item $T_i$ writes to a data item, and $T_j$ reads the same data item
\item Both $T_i$ and $T_j$ write to the same data item
\end{itemize}
\vspace{-0.2cm}
Algorithms \ref{alg:conflict_check},\ref{alg:conflict_set_idenification_no_helper} and \ref{alg:conflict_set_idenification_helper} outline the process of detecting conflicts between two transactions.

\textbf{Bin Assignment:} In this stage, the objective is to assign bins to the transactions and ensure each bin comprises a list of transactions independent of each other. It means that any two transactions within the same bin do not conflict or update the same data items. The specific procedures for assigning bins in various MBPS frameworks are detailed in Algorithms \ref{alg:bin_num_assign_no_helper}, \ref{alg:calculate_bin_algo_no_helper}, \ref{alg:bin_num_assign_helper} and \ref{alg:calculate_bin_algo_helper}. The algorithms are designed to allocate transactions to bins as compactly as possible, thereby optimizing the number of bins.

\textbf{Transaction Execution:} In this stage, non-conflicting transactions within the same group are eligible for parallel execution once assigned to their respective bins. This framework streamlines the parallel execution process by adhering to an organized sequence, where transactions in Bin 1 are executed first, succeeded by executing Bin 2 transactions, and so forth. Algorithm \ref{alg:exec_transaction} details the process of selecting transactions bin by bin for execution.

\input{section-folder/P-code}

\subsection{Standard MBPS Framework}
\label{sec:standard}
The Standard MBPS framework incorporates barrier in its design that establishes a controlled synchronization point. This arrangement facilitates organized parallel execution of smart contracts. The framework's core revolves around two primary phases: conflict set identification and bin number assignment. These phases are distinctly separated by barriers, ensuring that upon the completion of conflict set identification by all threads, the subsequent phase of bin number assignment commences.

In the conflict set identification phase, transactions are dynamically allocated to sets based on their potential conflicts with other transactions. These conflicts, identified through analysis of input and output addresses, facilitate the detection of write-write, read-write, and write-read conflicts. This phase operates concurrently across multiple threads, optimizing processing resource utilization.

After the identification of conflict sets, the framework proceeds to assign bin numbers, whereby transactions are allocated to bins according to their respective conflict sets. This allocation ensures that conflicting transactions are not placed within the same bin, thereby enabling concurrent execution of transactions within each bin and enhancing system efficiency.

\subsection{Assisted MBPS Framework}
\label{sec:assisted}
The Assisted MBPS introduces a blockchain transaction execution framework that incorporates a barrier mechanism along with thread assistance to enhance efficiency. This variant enhances the standard MBPS framework by introducing helper threads to facilitate optimized execution, particularly when certain threads encounter crashing or latency issues.

Similar to the standard MBPS framework, this variant also comprises two distinct phases: conflict set identification and bin number assignment. However, these phases are augmented with the inclusion of helper threads. Upon completion of both phases, transactions can be executed in parallel, thereby further optimizing performance.

\subsection{Lockfree MBPS Frameowrk}
\label{sec:lockfree}
The Lockfree MBPS framework facilitates parallel transaction execution in a blockchain by leveraging lock-free data structures and multithreading, ensuring efficient processing without contention. Unlike previous variants, this framework adopts a lock-free approach, avoiding traditional locking mechanisms like mutexes or barriers for synchronization.

Similar to previous variants, this framework consists of two distinct phases. However, it distinguishes itself by utilizing helper threads and various atomic operations to achieve a lock-free approach, enabling efficient transaction allocation to bins without contention.

This framework employs a combination of Algorithms \ref{alg:conflict_set_idenification_helper}, \ref{alg:bin_num_assign_helper} and \ref{alg:calculate_bin_algo_helper} to identify transaction conflicts and allocate bin numbers to transactions. This implementation features an approach where each thread is responsible for claiming a transaction and executing the initial phase to detect conflicts. Upon completion of the first phase, threads seamlessly transition to the subsequent phase of assigning bins to transactions without necessitating synchronization through barriers or locks. Consequently, threads are not required to wait for the completion of other threads' first phase executions before proceeding to the second phase. This design ensures efficient and concurrent execution without the overhead associated with traditional synchronization mechanisms.

Overall, this framework significantly enhances blockchain system scalability and performance by enabling the parallel execution of transactions while maintaining data consistency and integrity.

\begin{lemma}
Given two transactions \( \text{txnA} \) and \( \text{txnB} \), where \( \text{txnA.writeList} \), \( \text{txnA.readList} \), \( \text{txnB.writeList} \), and \( \text{txnB.readList} \) represent the sets of data items read or written by each transaction, the function CheckConflicts(txnA, txnB) returns true if there is any overlap in their write sets or if there is a write-read or read-write conflict; otherwise, it returns false.
\end{lemma}

\begin{proof}
The correctness of the \textit{CheckConflicts} function is established by analyzing possible conflict scenarios between two transactions \( \text{txnA} \) and \( \text{txnB} \). A write-write conflict arises if there is at least one data item that both transactions write. Thus, the function should return true if the intersection of \( \text{txnA.writeList} \) and \( \text{txnB.writeList} \) is non-empty. Similarly, a read-write conflict occurs if \( \text{txnA} \) reads a data item that \( \text{txnB} \) writes, necessitating a true return if \( \text{txnA.readList} \) intersects with \( \text{txnB.writeList} \). Conversely, a write-read conflict occurs if \( \text{txnA} \) writes a data item that \( \text{txnB} \) reads, requiring a true return if \( \text{txnA.writeList} \) intersects with \( \text{txnB.readList} \). The function \textit{CheckConflicts} effectively covers all these scenarios, ensuring accurate detection of overlaps in write sets or write-read/read-write conflicts.
\end{proof}

\begin{lemma}
The function BinConflictSet iterates through transactions to identify conflicts and stores them in a conflictArray, accurately identifying conflicts between each transaction \( \text{txn}_i \) and all transactions \( \text{txn}_j \) where \( j < i \).
\end{lemma}

\begin{proof}
The \textit{BinConflictSet} function is designed to identify and record conflicts among transactions by evaluating each transaction \( \text{txn}_i \) against all preceding transactions \( \text{txn}_j \) (where \( j < i \)). Initially, \textit{conflictArray} is initialized with empty sets for each transaction. For each transaction \( \text{txn}_i \), an empty set named \textit{lowerConflicts} is initialized to collect conflicts. The function iterates through all previous transactions \( \text{txn}_j \) and uses `CheckConflicts` to determine if there are any conflicts between \( \text{txn}_i \) and \( \text{txn}_j \). If a conflict is detected, the index \( j \) is added to \textit{lowerConflicts}. After processing all preceding transactions, \textit{conflictArray[i]} is updated with the collected conflicts. This approach ensures that each transaction is compared with all earlier ones, and conflicts are accurately recorded.
\end{proof}

\begin{lemma}
The function BinConflictSetHelper identifies conflicts similarly to BinConflictSet, but uses helper threads for concurrent processing.
\end{lemma}

\begin{proof}
The \textit{BinConflictSetHelper} function operates similarly to \textit{BinConflictSet} but incorporates concurrent processing to enhance efficiency. Multiple threads are employed, each responsible for processing different transactions. Each thread retrieves its transaction index \( i \) using atomic fetch-and-add operations to ensure unique assignment. For the assigned transaction \( i \), each thread performs conflict detection by iterating through previous transactions \( j \) (where \( j < i \)) and utilizing \textit{CheckConflicts}. If conflicts are identified, they are added to \textit{lowerConflicts}. The function then updates \textit{conflictArray[i]} using compare-and-swap (CAS) operations to manage concurrent updates. Atomic operations ensure that updates to \textit{conflictArray} are accurately performed even with multiple threads involved, allowing \textit{BinConflictSetHelper} to correctly identify conflicts while benefiting from concurrent processing.
\end{proof}

\begin{lemma}
The function BinNumAssign assigns bin numbers to transactions based on their conflict sets without using helper threads.
\end{lemma}

\begin{proof}
The \textit{BinNumAssign} function assigns bin numbers to transactions based on their identified conflict sets without employing concurrent processing. For each transaction \( i \), the bin number is computed using the \textit{CalculateBin} function. The computed bin number, referred to as \textit{allotedBin}, is used to assign the transaction to a bin. The function employs atomic operations to update \textit{binArray} with the transaction index \( i \) at the \textit{allotedBin}. The update operation is retried until it succeeds, ensuring that the assignment is made correctly without conflicts. By performing bin assignments serially and utilizing atomic operations, \textit{BinNumAssign} guarantees that each transaction is correctly assigned to the appropriate bin based on its conflict set.
\end{proof}

\begin{lemma}
The function BinNumAssignHelper assigns bin numbers to transactions using helper threads for concurrent processing.
\end{lemma}

\begin{proof}
The \textit{BinNumAssignHelper} function assigns bin numbers to transactions similarly to \textit{BinNumAssign} but leverages concurrent processing through helper threads for improved performance. Multiple threads are utilized, each processing different transactions. Each thread obtains a transaction index \( i \) through atomic fetch-and-add operations to ensure unique processing. For each transaction \( i \), if the initial bin assignment is uninitialized (indicated by \( \text{initialBin}[i] = -1 \)), the thread calculates the bin number using \textit{CalculateBinHelper}. The thread then attempts to update \textit{binArray} with the transaction index \( i \) at the calculated bin number using CAS operations. This process is repeated until the update is successful, accommodating concurrent thread operations. By using atomic operations and helper threads, \textit{BinNumAssignHelper} efficiently and correctly assigns bin numbers to transactions while managing concurrent updates effectively.
\end{proof}

%% file: section-folder/P-code.tex
\begin{algorithm}
\caption{Function to Check Conflicts}
\label{alg:conflict_check}
\tiny
\algrenewcommand\alglinenumber[1]{\tiny #1:}
\begin{algorithmic}[1]
\Function{CheckConflicts}{$\text{\textit{txnA, txnB}}$}
    \If{$\text{\textit{txnA.writeList}} \cap \text{\textit{txnB.writeList}} \neq \phi$}
        \State \textbf{return true}

    \ElsIf{$\text{\textit{txnA.readList}} \cap \text{\textit{txnB.writeList}} \neq \phi$}
        \State \textbf{return true}
    \ElsIf{$\text{\textit{txnA.writeList}} \cap \text{\textit{txnB.readList}} \neq \phi$}
        \State \textbf{return true}
    \EndIf
\EndFunction
\end{algorithmic}
\end{algorithm}


Algorithm \ref{alg:conflict_check} verifies conflicts between two transactions (txnA and txnB) by scrutinizing their read and write sets. Any overlap between their write or read sets signifies potential data dependencies, triggering the function to return true, indicating a conflict.


\begin{algorithm}
\tiny
\caption{Conflict Set Identification Algorithm - Without Helper Threads}
\label{alg:conflict_set_idenification_no_helper}
\begin{multicols}{2}

\begin{algorithmic}[1]
\algrenewcommand\alglinenumber[1]{\tiny #1:}

\Function{BinConflictSet}{}
    
    \State $flag \gets 0$
    \State $i \gets \text{atomicFetchAdd}(i1, 1)$
        
    \While{$\text{i} < \left|\text{\textit{Txns}}\right| $}
        
        \State $\text{set<int>} *lowerConflicts$

        \If{$\left|\text{conflictArray}[i]\right| = \mathit{\emptyset}$}
            \State $txn A \gets Txns[i]$
            \For{$j \gets 0$ \textbf{to} $i-1$}
                \State $txn B \gets Txns[j]$
                \If{\Call{\textbf{CheckConflicts}}{\text{\textit{\textbf{A, B}}}} = 1}
                    \State $(*lowerConflicts).insert(j)$
                \EndIf
            \EndFor
            \State $\text{conflictArray}[i] \gets lowerConflicts$
        \EndIf
        \State $i \gets \text{atomicFetchAdd}(i1, 1)$
    \EndWhile
\EndFunction

\end{algorithmic}
\end{multicols}
\end{algorithm}


Algorithms \ref{alg:conflict_set_idenification_no_helper} and \ref{alg:conflict_set_idenification_helper}, namely \textit{BinConflictSet} and \textit{BinConflictSetHelper}, iterate through transactions to identify conflicts among them. These algorithms employ a \textit{conflictArray} to store conflict sets and a \textit{lowerConflicts} set to monitor conflicting transactions. Through comparison with preceding transactions, conflicts are identified and logged in the \textit{conflictArray}. In Algorithm \ref{alg:conflict_set_idenification_no_helper}, threads claim transactions and detect conflicts, while Algorithm \ref{alg:conflict_set_idenification_helper} introduces helper threads to ensure efficient conflict resolution, allowing faster threads to overtake slower ones.


\begin{algorithm}
\tiny
\caption{Conflict Set Identification Algorithm - Helper Threads}
\label{alg:conflict_set_idenification_helper}
\begin{multicols}{2}

\begin{algorithmic}[1]
\algrenewcommand\alglinenumber[1]{\tiny #1:}

\Function{BinConflictSetHelper}{}
    \State $conflictTxns \gets 0$
    \State $localCount \gets 0$
    \State $flag \gets 0$

    \While{$\text{\textit{conflictTxns}} < \left|\text{\textit{Txns}}\right| $}
        \State $i \gets \text{atomicFetchAdd}(i1, 1) \mod \left|\text{\textit{Txns}}\right| $
        
        \State $\text{set<int>} *lowerConflicts$

        \If{$\text{\textit{conflictArray}}[\textit{i}]$ \textit{is} $\mathit{\emptyset}$}
            \State $localCount \gets 0$
            \If{$flag = 1$}
                \State $\text{atomicFetchAdd}(threadCounter1, -1)$
            \EndIf
            \State $txn A \gets Txns[i]$
            \For{$j \gets 0$ \textbf{to} $i-1$}
                \State $txn B \gets Txns[j]$
                \If{\Call{\textbf{CheckConflicts}}{\text{\textit{\textbf{A, B}}}} = 1}
                    \State $(*lowerConflicts).insert(j)$
                \EndIf
            \EndFor
            \State $localConf \gets lowerConflicts$
            \State $temp \gets \text{NULL}$

            \If{$\text{conflictArray}[i].\text{CAS}(temp, localConf)$}
                \State $\text{atomicFetchAdd}(conflictTxns, 1)$
            \EndIf
        \Else
            \State $localCount \gets localCount + 1$

            \If{$localCount = \left|\text{\textit{Txns}}\right| \text{ and } flag = 0$}
                \State $flag \gets 1$
                \State $\text{atomicFetchAdd}(threadCounter1, 1)$
            \EndIf
        \EndIf

        \If{$(\text{threadCounter1} = \text {numThreads})$ \textbf{or} $(localCount = \left|\text{\textit{Txns}}\right|)$}
            \State $\text{atomicStore}(\textit{conflictTxns}, \left|\text{\textit{Txns}}\right|)$
        \EndIf
    \EndWhile
\EndFunction

\end{algorithmic}
\end{multicols}
\end{algorithm}


\begin{algorithm}
\tiny
\caption{Bin Number Assignment Algorithm - Without Helper Threads}
\label{alg:bin_num_assign_no_helper}
\begin{multicols}{2}
\algrenewcommand\alglinenumber[1]{\tiny #1:}
\begin{algorithmic}[1]
    \Function{BinNumAssign}{}
        \State $flag \gets 0$
        \State $i \gets \text{atomicFetchAdd}(i2, 1)$
        
        \While{$\text{i} < \left|\text{\textit{Txns}}\right| $}
                
                \State $allotedBin \gets \Call {\text{\textbf{CalculateBin}}}{\textit{\textbf{i}}}$
                
                \State $localVal \gets \text{allotedBin}$
                \State $\text{set<int>} *copy1, *copy2, *tempCopy$
                
                \Repeat
                    \State $Copy1 \gets \text{binArray}[\text{allotedBin}]$
                    
                    \If{$\text{copy1} = \text{NULL}$}
                        \State $(*tempCopy).insert(i)$
                        \State $\text{copy2} \gets tempCopy$
                    \Else
                        \If{$i \in *copy1$}
                            \State \textbf{break}
                        \EndIf
                        
                        \For{$\text{a} \in *copy1$}
                            \State $(*\text{copy2}).insert(\text{a})$
                        \EndFor
                        
                        \State $(*\text{copy2}).insert(i)$
                    \EndIf
                \Until{$\text{binArray}[\text{allotedBin}].\text{CAS}(\text{copy1}, \text{copy2})$}

                \State $\text{initialBin}[i] \gets \text{localVal}$
                \State $i \gets \text{atomicFetchAdd}(i2, 1)$
                
        \EndWhile

    \EndFunction
\end{algorithmic}
\end{multicols}
\end{algorithm}


\begin{algorithm}
\tiny
\caption{Calculate Bin Algorithm - Without Helper Threads}
\label{alg:calculate_bin_algo_no_helper}
\algrenewcommand\alglinenumber[1]{\tiny #1:}
\begin{algorithmic}[1]
    \Function{CalculateBin}{$i$}
        \State $\text{currentBin} \gets -1$
        \If{$|\text{conflictArray}[i]| \neq \mathit{\emptyset}$}
            \For{$\text{conflictTxn} \in \text{conflictArray}[i]$}
                \While{$\text{initialBin}[\text{conflictTxn}] = -1$}
                    \State \text{continue}
                \EndWhile
                \If{$\text{initialBin}[\text{conflictTxn}] > \text{currentBin}$}
                    \State $\text{currentBin} \gets \max(\text{currentBin}, \text{initialBin}[\text{conflictTxn}])$
                \EndIf
            \EndFor
        \EndIf
        \State $\text{currentBin} \gets \text{currentBin} + 1$
        \State \Return $\text{currentBin}$
        
    \EndFunction
\end{algorithmic}
\end{algorithm}


Algorithm \ref{alg:bin_num_assign_no_helper} and Algorithm \ref{alg:bin_num_assign_helper} are tailored for assigning bin numbers to transactions within a concurrent environment. Algorithm \ref{alg:bin_num_assign_no_helper} iterates through transactions, computing the bin number for each transaction using Algorithm \ref{alg:calculate_bin_algo_no_helper}, and subsequently assigns the transaction to the corresponding bin. It leverages atomic operations for thread safety and guarantees no transaction is assigned to a bin until all its dependencies are resolved. Algorithm \ref{alg:bin_num_assign_helper} enhances concurrency by employing helper threads. It operates similarly to Algorithm \ref{alg:bin_num_assign_no_helper} but with additional logic for efficient handling of helper threads. Helper threads aid in processing transactions concurrently, potentially speeding up the assignment process.

Algorithm \ref{alg:calculate_bin_algo_no_helper} calculates the bin number for a given transaction based on conflicts with other transactions. It identifies conflicts and assigns the transaction to the next available bin. Algorithm \ref{alg:calculate_bin_algo_helper} shares a similar logic with Algorithm \ref{alg:calculate_bin_algo_no_helper} for bin number calculation but incorporates modifications to manage concurrent execution using helper threads.

Algorithm \ref{alg:exec_transaction}, ExecuteTransaction, selects transactions bin by bin and forwards them to the scheduler for execution.


\begin{algorithm}
\tiny
\caption{Bin Number Assignment Algorithm - Helper Threads}
\label{alg:bin_num_assign_helper}
\begin{multicols}{2}
\algrenewcommand\alglinenumber[1]{\tiny #1:}
\begin{algorithmic}[1]
    \Function{BinNumAssignHelper}{}
        \State $processedTxns \gets 0$
        \State $localCount \gets 0$
        \State $flag \gets 0$
        
        \While{$\text{processedTxns} < \left|\text{\textit{Txns}}\right| $}
            \State $i \gets \text{atomicFetchAdd}(i2, 1) \mod \left|\text{\textit{Txns}}\right|$
            
            \If{$\text{initialBin}[i] = -1$}
                \State $localCount \gets 0$
                
                \If{$\text{flag} = 1$}
                    \State $\text{atomicFetchAdd}(\text{threadCounter2}, -1)$
                \EndIf
                
                \State $allotedBin \gets \Call {\text{\textbf{CalculateBinHelper}}}{\textit{\textbf{i}}}$
                
                \If{$\text{allotedBin} = -1$}
                    \State \textbf{continue}
                \EndIf
                
                \State $localVal \gets \text{allotedBin}$
                \State $\text{set<int>} *copy1, *copy2, *tempCopy$
                
                \Repeat
                    \State $Copy1 \gets \text{binArray}[\text{allotedBin}]$
                    
                    \If{$\text{copy1} = \text{NULL}$}
                        \State $(*tempCopy).insert(i)$
                        \State $\text{copy2} \gets tempCopy$
                    \Else
                        \If{$i \in *copy1$}
                            \State \textbf{break}
                        \EndIf
                        
                        
                        \For{$\text{a} \in *copy1$}
                            \State $(*\text{copy2}).insert(\text{a})$
                        \EndFor
                        
                        \State $(*\text{copy2}).insert(i)$
                    \EndIf
                \Until{$\text{binArray}[\text{allotedBin}].\text{CAS}(\text{copy1}, \text{copy2})$}
                
                \State $temp1 \gets -1$
                
                \If{$\text{initialBin}[i].\text{CAS}(\text{temp1}, \text{localVal})$}
                    \State $\text{atomicFetchAdd}(\text{processedTxns}, 1)$
                \EndIf
            \Else
                \State $localCount \gets \text{localCount} + 1$
                
                \If{$\text{localCount} = \left|\text{\textit{Txns}}\right|$ \textbf{and} $\text{flag} = 0$}
                    \State $\text{flag} \gets 1$
                    \State $\text{atomicFetchAdd}(\text{threadCounter2}, 1)$
                \EndIf
            \EndIf
            
            \If{$(\text{threadCounter2} = \text{numThreads})$ \textbf{or} $(\text{localCount} = \left|\text{\textit{Txns}}\right|)$}
                \State $\text{atomicStore}(\text{processedTxns}, \left|\text{\textit{Txns}}\right|)$
            \EndIf
        \EndWhile

    \EndFunction
\end{algorithmic}
\end{multicols}
\end{algorithm}


\begin{algorithm}
\tiny
\caption{Calculate Bin Algorithm - Helper Threads}
\label{alg:calculate_bin_algo_helper}
\algrenewcommand\alglinenumber[1]{\tiny #1:}
\begin{algorithmic}[1]
    \Function{CalculateBinHelper}{$i$}
        \State $\text{currentBin} \gets -1$
        
        \If{$\left|\text{conflictArray}[i]\right| \neq \mathit{\emptyset}$}
            \For{$\text{conflictTxn} \in \text{conflictArray}[i]$}
                \If{$\text{initialBin}[\text{conflictTxn}] = -1$}
                    \State \textbf{return} $-1$
                \Else
                    \If{$\text{initialBin}[\text{conflictTxn}] > \text{currentBin}$}
                        \State $\text{currentBin} \gets \max(\text{currentBin}, \text{initialBin}[\text{conflictTxn}]$
                    \EndIf
                \EndIf
            \EndFor
        \EndIf
        
        \State $\text{currentBin} \gets \text{currentBin} + 1$
        \State \textbf{return} $\text{currentBin}$
    \EndFunction
\end{algorithmic}
\end{algorithm}


\begin{algorithm}
\tiny
\caption{Execute Transaction}
\begin{multicols}{2}
\algrenewcommand\alglinenumber[1]{\tiny #1:}
\label{alg:exec_transaction}
\begin{algorithmic}[1]
\Function{ExecuteTransaction}{$\text{currBin, currTrans}$}
    \If{$\text{glbptr} \geq 0$}
        \If{$\text{currBin} > \text{glbptr}$}
            \State \Return $-1$
        \EndIf
        \State $val \gets \text{totalTransBin}[\text{currBin}]$
        \If{$\text{currTrans} \geq \text{val}$}
            \State \Return $-1$
        \EndIf
        \If{$\text{currTrans} < \text{val}$}
            \State \Return $\text{binMatrix}[\text{currBin}][\text{currTrans}]$
        \Else
            \State \Return $-1$
        \EndIf
    \EndIf
    \State \Return $-1$
\EndFunction
\end{algorithmic}
\end{multicols}
\end{algorithm}

%% file: section-folder/Experiments.tex
In this section, we present a comprehensive analysis of the experiments conducted to evaluate the performance of our framework integrated with the Hyperledger Sawtooth \cite{sawtooth:url} blockchain. We chose Hyperledger Sawtooth as the testing platform due to its robust support for parallelism and the presence of an inbuilt parallel scheduler. We made modifications to the \textit{scheduler\_parallel.py} file, the parallel scheduler module of Hyperledger Sawtooth v1.2.6. Although the Sawtooth framework is developed in Python, we developed our multi-threaded MBPS framework in C++. We selected C++ for its provision of low-level control over parallelism through features such as threads, mutexes, condition variables, and atomic operations provided by the libraries.

We conducted the experiments on a machine featuring an x86\_64 architecture with 56 CPUs, 2 threads per core, and 14 cores per socket (Intel Xeon CPU E5-2690 v4 @ 2.60GHz). We compared the performance of three versions of our MBPS framework with the inbuilt parallel tree scheduler of Sawtooth and the serial scheduler of Sawtooth. Additionally, we compared the results with the \textit{ADJ\_DAG} and \textit{LL\_DAG} frameworks \cite{10.1007/978-3-031-39698-4_13}.

We conducted three distinct types of experiments to assess different aspects of the framework's performance under varying conditions: Baseline Performance Evaluation, Threads Latency Impact Analysis, and Threads Crash Resilience Analysis.
We have employed three distinct conflict parameters (CP) outlined in paper \cite{10.1007/978-3-031-39698-4_13} to assess the performance of the experiments discussed above. The first parameter, CP1, indicates the percentage of transactions that contain at least one dependency. The second parameter, CP2, represents the percentage of dependent transactions in relation to the total number of transactions. Lastly, CP3 measures the percentage of disjoint transactions relative to the total number of transactions.

\vspace{-0.1cm}

\subsection{Baseline Performance Evaluation}
This experiment serves as a benchmark, providing a standard for expected performance without additional factors such as delays or crashes. We conducted this experiment across all frameworks to facilitate comparative analysis. \grfref{baseline-graphs} illustrates the performance comparison of all frameworks.

In \grfref{baseline-graphs}(a), we varied the number of transactions and observed the execution times. We noted that when the number of transactions was low, all frameworks exhibited similar performance. However, as the number of transactions increases, the performance of the serial and tree schedulers deteriorates, while the remaining frameworks maintain relatively comparable execution time.

In \grfref{baseline-graphs}(b) and \grfref{baseline-graphs}(c), we recorded the throughput of all frameworks while varying the number of transactions and dependency percentage, respectively. Similar to the observations in Graph 1a, the results indicated that serial and tree schedulers deviated from the norm. Conversely, the performance of the other frameworks remained comparable across the board.
\input{graph-folder/threads-baseline-graph}

\vspace{-1cm}

\subsection{Threads Latency Impact Analysis}
This experiment focuses on evaluating the impact of introducing delays in threads on performance. We conducted the experiments for both the MBPS and DAG frameworks to compare results and understand how performance is affected under such conditions. \grfref{delay-graphs} presents the performance comparison of the different frameworks.

In \grfref{delay-graphs}(a) and \grfref{delay-graphs}(c), while keeping 600 transactions fixed, we intentionally delayed certain threads to assess the performance of the frameworks. As we increased the percentage of delayed threads, we observed that Lockfree MBPS frameworks began to outperform others, followed by the Assisted MBPS framework. In \grfref{delay-graphs}(b), we increased the number of transactions while maintaining one-third of the threads delayed across all cases. Here as well, we found that Lockfree and Assisted MBPS performed better due to the presence of helper threads.
\input{graph-folder/threads-delay-graph}

\vspace{-1cm}

\subsection{Threads Crash Resilience Analysis}
This experiment evaluates the performance of the framework in the event of threads crashing. It was specifically conducted on the lockfree bin scheduler, as other schedulers were lock-based and did not incorporate mechanisms to handle thread crashes. \grfref{crash-graphs} showcases the performance of the lockfree bin scheduler in case of threads crashing.

In \grfref{crash-graphs}, while maintaining a fixed number of 600 transactions, we intentionally crashed a few threads to evaluate the performances of the frameworks. All the other frameworks, except Lockfree MBPS, were unable to execute the transactions as they are not thread-crash-tolerant algorithms. We varied the number of crashed threads from 1 to 99\%, and observed that Lockfree MBPS completed the execution, although it took more time. It completed its execution even when all the other threads except one were crashed.
\input{graph-folder/threads-crash-graph}

%% file: graph-folder/threads-baseline-graph.tex
\begin{center}
\begin{Graph}[!ht]
\begin{minipage}{0.33\textwidth}
\vspace{0.85cm}
\pgfplotsset{width=4.5cm,compat=1.9}
\centering
\begin{tikzpicture}
\label{graph:baseline1}
\begin{axis}[
    title = {(a) CP1},
    title style={at={(0.5,-0.5)},anchor=north}, 
    xlabel={\tiny No. of Transactions},
    ylabel={\tiny Execution Time [seconds]},
    xtick={200,400,600,800,1000,1200},
    xticklabel style={font=\tiny, /pgf/number format/1000 sep= , rotate=15},
    yticklabel style={font=\tiny, /pgf/number format/1000 sep= },
    y label style={yshift=-5pt},
    legend style={at={(0.5,1.2)},
    anchor=north,legend columns=-1},
    ymajorgrids=true,
    grid style=dashed,
]

    \addplot[
    color=orange,
    mark=square*,
    ]
    coordinates {
    (200, 59.4880723953247)
    (400, 133.108985424041)
    (600, 208.619644641876)
    (800, 287.83289194107)
    (1000, 414.000573158264)
    (1200, 496.014652252197)
    };

    \addplot[
    color=violet,
    mark=star,
    ]
    coordinates {
    (200, 55.5977654457092)
    (400, 113.712277412414)
    (600, 175.368566513061)
    (800, 242.035264968872)
    (1000, 394.304020404815)
    (1200, 477.355322837829)
    };

    \addplot[
    color=teal,
    mark=o,
    ]
    coordinates {
    (200, 51.0700440406799)
    (400, 100.923318862915)
    (600, 165.826790332794)
    (800, 230.152275562286)
    (1000, 350.475747585296)
    (1200, 426.033885478973)
    };

    \addplot[
    color=blue,
    mark=*,
    ]
    coordinates {
    (200, 53.5197877883911)
    (400, 106.391243934631)
    (600, 165.559518337249)
    (800, 222.098581790924)
    (1000, 354.323296546936)
    (1200, 432.55479812622)
    };

    \addplot[
    color=brown,
    mark=triangle,
    ]
    coordinates {
    (200, 51.0003218650817)
    (400, 104.1633374691)
    (600, 160.447205781936)
    (800, 225.73484134674)
    (1000, 337.504456043243)
    (1200, 421.768282175064)
    };

    \addplot[
    color=green,
    mark=Mercedes star flipped,
    ]
    coordinates {
    (200, 51.0022809505462)
    (400, 105.822781801223)
    (600, 164.263003587722)
    (800, 225.731611728668)
    (1000, 348.164312839508)
    (1200, 432.820511817932)
    };

    \addplot[
    color=red,
    mark=square,
    ]
    coordinates {
    (200, 50.2075493335723)
    (400, 103.8509683609)
    (600, 162.928800582885)
    (800, 225.915188074111)
    (1000, 349.940569639205)
    (1200, 423.27399778366)
    };

\end{axis}
\end{tikzpicture}
\end{minipage}%
\hfill
\begin{minipage}{0.33\textwidth}
\pgfplotsset{width=4.5cm,compat=1.9}
\begin{tikzpicture}
\label{graph:baseline2}
\begin{axis}[
    title = {(b) CP2},
    title style={at={(0.5,-0.5)},anchor=north}, 
    xlabel={\tiny No. of Transactions},
    ylabel={\tiny Throughput [txns/s]},
    xtick={200,400,600,800,1000,1200},
    xticklabel style={font=\tiny, /pgf/number format/1000 sep= , rotate=15},
    yticklabel style={font=\tiny, /pgf/number format/1000 sep= },
    y label style={yshift=-5pt}, 
    legend style={at={(0.5,1.5)},
    anchor=north,legend columns=-1},
    ymajorgrids=true,
    grid style=dashed,
]

    \addplot[
    color=orange,
    mark=square*,
    ]
    coordinates {
    (200, 54.6069589663134)
    (400, 52.9997514158226)
    (600, 52.180634879544)
    (800, 52.0664565125153)
    (1000, 48.9204594653505)
    (1200, 47.6709517483711)
    };

    \addplot[
    color=violet,
    mark=star,
    ]
    coordinates {
    (200, 62.0806171023257)
    (400, 58.3289379461728)
    (600, 57.1279992579621)
    (800, 54.8600141011671)
    (1000, 51.7899572820621)
    (1200, 50.0298782804865)
    };

    \addplot[
    color=teal,
    mark=o,
    ]
    coordinates {
    (200, 62.3126494052537)
    (400, 60.4219683842352)
    (600, 59.4772999253517)
    (800, 59.0595125174713)
    (1000, 57.5875740679185)
    (1200, 58.0835748306135)
    };

    \addplot[
    color=blue,
    mark=*,
    ]
    coordinates {
    (200, 61.3359467708669)
    (400, 60.4806493253306)
    (600, 59.3849992103105)
    (800, 56.8526817201439)
    (1000, 56.7539512277498)
    (1200, 55.5454929012194)
    };

    \addplot[
    color=brown,
    mark=triangle,
    ]
    coordinates {
    (200, 63.7069336211917)
    (400, 60.9941457498684)
    (600, 60.8652791557811)
    (800, 57.829607127108)
    (1000, 56.9780261761175)
    (1200, 55.7457897565487)
    };

    \addplot[
    color=green,
    mark=Mercedes star flipped,
    ]
    coordinates {
    (200, 63.5094431897819)
    (400, 62.0813118875981)
    (600, 60.270438240771)
    (800, 58.8789908743268)
    (1000, 56.4620785222196)
    (1200, 56.5841572271394)
    };

    \addplot[
    color=red,
    mark=square,
    ]
    coordinates {
    (200, 62.6447261425672)
    (400, 61.2778774890498)
    (600, 59.3982740925502)
    (800, 58.4077989383728)
    (1000, 57.3064647077269)
    (1200, 56.1717480864563)
    };

    \legend{\tiny Serial, \tiny Tree, \tiny Adj\_DAG, \tiny LL\_DAG, \tiny Standard, \tiny Assisted, \tiny Lockfree}{\hspace{-3.4 cm}}

\end{axis}
\end{tikzpicture}
\end{minipage}%
\hfill
\begin{minipage}{0.33\textwidth}
\vspace{1.1cm}
\pgfplotsset{width=4.5cm,compat=1.9}
\begin{tikzpicture}
\label{graph:baseline3}
\begin{axis}[
    title={(c) CP3},
    title style={at={(0.5,-0.5)},anchor=north}, 
    xlabel={\tiny Dependency Percentage},
    ylabel={\tiny Throughput [txns/s]},
    xtick={0,20,40,60,80,100},
    xticklabel style={font=\tiny, /pgf/number format/1000 sep= },
    yticklabel style={font=\tiny, /pgf/number format/1000 sep= },
    y label style={yshift=-4pt},
    legend style={at={(0.5,1.2)},
    anchor=north,legend columns=-1},
    ymajorgrids=true,
    grid style=dashed,
]

    \addplot[
    color=orange,
    mark=square*,
    ]
    coordinates {
    (0, 48.2202880881057)
    (20, 50.9503837212725)
    (40, 49.50529365798)
    (60, 46.8068635962329)
    (80, 45.5494176474553)
    (100, 42.7642844635524)
    };

    \addplot[
    color=violet,
    mark=star,
    ]
    coordinates {
    (0, 64.5281673227953)
    (20, 54.8983304290748)
    (40, 51.8372937295255)
    (60, 49.0843911404695)
    (80, 50.0282721551632)
    (100, 49.7503587599128)
    };

    \addplot[
    color=teal,
    mark=o,
    ]
    coordinates {
    (0, 67.6022277622984)
    (20, 58.6849001082319)
    (40, 56.22950399466)
    (60, 55.997255932809)
    (80, 55.6716115243087)
    (100, 56.2492810514448)
    };

    \addplot[
    color=blue,
    mark=*,
    ]
    coordinates {
    (0, 70.9334965576199)
    (20, 57.0684833937445)
    (40, 56.8295611641)
    (60, 56.5143114073968)
    (80, 55.4538662858671)
    (100, 58.3527317641228)
    };

    \addplot[
    color=brown,
    mark=triangle,
    ]
    coordinates {
    (0, 70.8131158799343)
    (20, 58.4111276129311)
    (40, 58.222314089468)
    (60, 57.2976007386501)
    (80, 58.1614327245399)
    (100, 57.6099347557158)
    };

    \addplot[
    color=green,
    mark=Mercedes star flipped,
    ]
    coordinates {
    (0, 68.5099223860074)
    (20, 57.6163767623198)
    (40, 56.6601092116857)
    (60, 58.7870675081018)
    (80, 58.0570896163345)
    (100, 58.2893359205743)
    };

    \addplot[
    color=red,
    mark=square,
    ]
    coordinates {
    (0, 66.7233738709988)
    (20, 56.8693657300459)
    (40, 56.6234569877998)
    (60, 56.9067174683176)
    (80, 56.7904835932267)
    (100, 56.6551592313122)
    };

\end{axis}
\end{tikzpicture}
\end{minipage}%

\vspace{-0.2cm}
\caption{Simple Wallet Smart Contracts - Baseline Performance Analysis}
\label{graph:baseline-graphs}
\end{Graph}
\end{center}

%% file: graph-folder/threads-delay-graph.tex
\begin{center}
\begin{Graph}[!ht]
\begin{minipage}{0.33\textwidth}
\vspace{0.85cm}
\pgfplotsset{width=4.5cm,compat=1.9}
\centering
\begin{tikzpicture}
\label{graph:delay1}
\begin{axis}[
    title = {(a) 600 txns, CP1},
    title style={at={(0.5,-0.5)},anchor=north}, 
    xlabel={\tiny Delayed Threads [\%]},
    ylabel={\tiny Execution Time [seconds]},
    xtick={0,20,40,60,80,100},
    xticklabel style={font=\tiny, /pgf/number format/1000 sep= },
    yticklabel style={font=\tiny, /pgf/number format/1000 sep= },
    y label style={yshift=-5pt},
    legend style={at={(0.5,1.2)},
    anchor=north,legend columns=-1},
    ymajorgrids=true,
    grid style=dashed,
]

    \addplot[
    color=teal,
    mark=o,
    ]
    coordinates {
    (0, 168.057718276977)
    (20, 502.660870552063)
    (40, 521.651227474212)
    (60, 535.395941734314)
    (80, 545.32303571701)
    };

    \addplot[
    color=blue,
    mark=*,
    ]
    coordinates {
    (0, 166.712996482849)
    (20, 277.883512973785)
    (40, 352.251249790191)
    (60, 418.476637363433)
    (80, 525.300290584564)
    };

    \addplot[
    color=brown,
    mark=triangle,
    ]
    coordinates {
    (0, 165.612473487854)
    (20, 316.410882472991)
    (40, 469.641005992889)
    (60, 618.153278827667)
    (80, 775.592958927154)
    };

    \addplot[
    color=green,
    mark=Mercedes star flipped,
    ]
    coordinates {
    (0, 167.280564308166)
    (20, 267.520802021026)
    (40, 346.495583057403)
    (60, 401.292874813079)
    (80, 506.148881912231)
    };

    \addplot[
    color=red,
    mark=square,
    ]
    coordinates {
    (0, 164.539890289306)
    (20, 261.311967372894)
    (40, 340.995609760284)
    (60, 403.686907768249)
    (80, 490.216394901275)
    };


\end{axis}
\end{tikzpicture}
\end{minipage}%
\hfill
\begin{minipage}{0.33\textwidth}
\pgfplotsset{width=4.5cm,compat=1.9}
\begin{tikzpicture}
\label{graph:delay2}
\begin{axis}[
    title = {(b) 33\% delay ratio, CP2},
    title style={at={(0.5,-0.5)},anchor=north}, 
    xlabel={\tiny No. of Transactions},
    ylabel={\tiny Execution Time [seconds]},
    xtick={200,400,600,800,1000,1200},
    xticklabel style={font=\tiny, /pgf/number format/1000 sep= , rotate=15},
    yticklabel style={font=\tiny, /pgf/number format/1000 sep= },
    y label style={yshift=-6pt}, 
    legend style={at={(0.5,1.5)},
    anchor=north,legend columns=-1},
    ymajorgrids=true,
    grid style=dashed,
]

    \addplot[
    color=teal,
    mark=o,
    ]
    coordinates {
    (200, 190.068962574005)
    (400, 371.348848342895)
    (600, 552.187280654907)
    (800, 740.154435634613)
    (1000, 946.777036190033)
    (1200, 1134.57604646682)
    };

    \addplot[
    color=blue,
    mark=*,
    ]
    coordinates {
    (200, 143.5405779)
    (400, 269.7447515)
    (600, 381.1887455)
    (800, 495.95155)
    (1000, 629.0836506)
    (1200, 772.7933121)
    };

    \addplot[
    color=brown,
    mark=triangle,
    ]
    coordinates {
    (200, 152.806134223937)
    (400, 309.697196483612)
    (600, 466.570508480072)
    (800, 612.388832569122)
    (1000, 765.667107105255)
    (1200, 941.159372329711)
    };

    \addplot[
    color=green,
    mark=Mercedes star flipped,
    ]
    coordinates {
    (200, 141.497933864593)
    (400, 252.167470455169)
    (600, 360.960683822631)
    (800, 476.924703121185)
    (1000, 581.791539192199)
    (1200, 712.028737068176)
    };

    \addplot[
    color=red,
    mark=square,
    ]
    coordinates {
    (200, 132.920138835906)
    (400, 240.765206813812)
    (600, 351.821790218353)
    (800, 464.786880016326)
    (1000, 577.140250205993)
    (1200, 703.353397846221)
    };

    \legend{\tiny Adj\_DAG, \tiny LL\_DAG, \tiny Standard, \tiny Assisted, \tiny Lockfree}{\hspace{-1.95cm}}

\end{axis}
\end{tikzpicture}
\end{minipage}%
\hfill
\begin{minipage}{0.33\textwidth}
\vspace{1.1cm}
\pgfplotsset{width=4.5cm,compat=1.9}
\begin{tikzpicture}
\label{graph:delay3}
\begin{axis}[
    title={(c) 600 txns, CP3},
    title style={at={(0.5,-0.5)},anchor=north}, 
    xlabel={\tiny Delayed Threads [\%]},
    ylabel={\tiny Throughput [txns/s]},
    xtick={0,20,40,60,80,100},
    xticklabel style={font=\tiny, /pgf/number format/1000 sep= },
    yticklabel style={font=\tiny, /pgf/number format/1000 sep= },
    y label style={yshift=-4pt},
    legend style={at={(0.5,1.2)},
    anchor=north,legend columns=-1},
    ymajorgrids=true,
    grid style=dashed,
]

    \addplot[
    color=teal,
    mark=o,
    ]
    coordinates {
    (0, 58.273911924894)
    (20, 22.454474196583)
    (40, 20.9327577658072)
    (60, 19.1028037220802)
    (80, 18.0275915706472)
    };

    \addplot[
    color=blue,
    mark=*,
    ]
    coordinates {
    (0, 61.6877623612721)
    (20, 39.337698030942)
    (40, 31.6522957287906)
    (60, 25.3394824983952)
    (80, 21.1260522433104)
    };

    \addplot[
    color=brown,
    mark=triangle,
    ]
    coordinates {
    (0, 60.3798276653232)
    (20, 33.3156169953172)
    (40, 23.5484575362366)
    (60, 17.7904073319895)
    (80, 14.5583021821032)
    };

    \addplot[
    color=green,
    mark=Mercedes star flipped,
    ]
    coordinates {
    (0, 58.2264905082823)
    (20, 39.5252029666192)
    (40, 31.7611855140392)
    (60, 26.6360088992369)
    (80, 22.4896767251689)
    };

    \addplot[
    color=red,
    mark=square,
    ]
    coordinates {
    (0, 59.7432488908595)
    (20, 41.5750902015962)
    (40, 32.4368222914138)
    (60, 27.2615097289864)
    (80, 23.0398940373407)
    };

\end{axis}
\end{tikzpicture}
\end{minipage}%

\vspace{-0.2cm}
\caption{Simple Wallet Smart Contracts - Threads Latency Analysis}
\label{graph:delay-graphs}
\end{Graph}
\end{center}

%% file: graph-folder/threads-crash-graph.tex
\begin{center}
\begin{Graph}[!ht]
\begin{minipage}{0.33\textwidth}
\vspace{0.85cm}
\pgfplotsset{width=4.5cm,compat=1.9}
\centering
\begin{tikzpicture}
\label{graph:crash1}
\begin{axis}[
    title = {(a) 600 txns, CP1},
    title style={at={(0.5,-0.5)},anchor=north}, 
    xlabel={\tiny Threads Crashing [\%]},
    ylabel={\tiny Execution Time [seconds]},
    xtick={0,20,40,60,80,100},
    xticklabel style={font=\tiny, /pgf/number format/1000 sep= },
    yticklabel style={font=\tiny, /pgf/number format/1000 sep= },
    y label style={yshift=-6pt},
    legend style={at={(0.5,1.2)},
    anchor=north,legend columns=-1},
    ymajorgrids=true,
    grid style=dashed,
]

    \addplot[
    color=red,
    mark=square,
    ]
    coordinates {
    (0,163.612132072448)
    (20,365.551707744598)
    (40,525.877628326416)
    (60,657.422013282775)
    (80,819.843959808349)
    (99,964.266777038574)
    };

\end{axis}
\end{tikzpicture}
\end{minipage}%
\hfill
\begin{minipage}{0.33\textwidth}
\pgfplotsset{width=4.5cm,compat=1.9}
\begin{tikzpicture}
\label{graph:crash2}
\begin{axis}[
    title = {(b) 33\% crash ratio, CP2},
    title style={at={(0.5,-0.5)},anchor=north}, 
    xlabel={\tiny No. of Transactions},
    ylabel={\tiny Execution Time [seconds]},
    xtick={200,400,600,800,1000,1200},
    xticklabel style={font=\tiny, /pgf/number format/1000 sep= , rotate=15},
    yticklabel style={font=\tiny, /pgf/number format/1000 sep= },
    y label style={yshift=-6pt}, 
    legend style={at={(0.5,1.5)},
    anchor=north,legend columns=-1},
    ymajorgrids=true,
    grid style=dashed,
]

    \addplot[
    color=red,
    mark=square,
    ]
    coordinates {
    (200,223.774311542511)
    (400,366.06077671051)
    (600,514.296109676361)
    (800,667.628016471862)
    (1000,822.176465988159)
    (1200,986.234416961669)
    };

    \legend{\tiny Lockfree}{\hspace{0 cm}}

\end{axis}
\end{tikzpicture}
\end{minipage}%
\hfill
\begin{minipage}{0.33\textwidth}
\vspace{1.1cm}
\pgfplotsset{width=4.5cm,compat=1.9}
\begin{tikzpicture}
\label{graph:crash3}
\begin{axis}[
    title={(c) 600 txns, CP3},
    title style={at={(0.5,-0.5)},anchor=north}, 
    xlabel={\tiny Threads Crashing [\%]},
    ylabel={\tiny Throughput [txns/s]},
    xtick={0,20,40,60,80,100},
    xticklabel style={font=\tiny, /pgf/number format/1000 sep= },
    yticklabel style={font=\tiny, /pgf/number format/1000 sep= },
    y label style={yshift=-4pt},
    legend style={at={(0.5,1.2)},
    anchor=north,legend columns=-1},
    ymajorgrids=true,
    grid style=dashed,
]

    \addplot[
    color=red,
    mark=square,
    ]
    coordinates {
    (0,60.1248036809655)
    (20,28.1937586876562)
    (40,21.3184126781384)
    (60,16.7063740686477)
    (80,13.9203847602228)
    (99,11.9561828536549)
    };

\end{axis}
\end{tikzpicture}
\end{minipage}%

\vspace{-0.2cm}
\caption{Simple Wallet Smart Contracts - Threads Crashing Analysis}
\label{graph:crash-graphs}
\end{Graph}
\end{center}

%% file: section-folder/Conclusion-Future-Work.tex
In this paper, we present the MBPS framework, developed to parallelize blockchain smart contract transactions and harness the capabilities of multicore systems. Our framework consists of three variants: Standard MBPS, Assisted MBPS, and Lockfree MBPS, each offering distinct features to enhance transaction execution efficiency while maintaining deterministic order.
Through experiments, we assessed the performance of our framework against existing parallel execution frameworks on the Hyperledger Sawtooth blockchain platform. The results demonstrated significant improvements in throughput and execution time, particularly in scenarios involving high transaction volumes, latency, and thread crashes.
Our findings indicate that the Lockfree MBPS framework, in particular, excels in resilience to thread crashes and efficient transaction processing without contention. Additionally, the Assisted MBPS framework shows promising results in mitigating latency issues through the introduction of helper threads.

In the future, our objective is to expand the proposed MBPS framework into distributed settings, offering a promising direction for both research and development. This extension will involve designing the algorithm to operate efficiently in distributed environments, addressing challenges such as network latency, communication overhead, and synchronization across multiple nodes.
Additionally, we will explore further optimizations and enhancements to refine the MBPS framework. This will involve fine-tuning algorithms for conflict detection and bin assignment stages to improve efficiency and reduce overhead.
Furthermore, we are investigating the integration of the MBPS framework with emerging technologies such as machine learning for dynamic scheduling and resource allocation. This integration holds the potential to unlock even greater scalability and performance within blockchain ecosystems.